\documentclass[envcountsame]{llncs}
\usepackage{tabularx}
\usepackage{amsmath}
\usepackage{times}
\usepackage{helvet}
\usepackage{courier}
\usepackage{ifpdf}
\ifpdf
\usepackage{pdfsync}
\fi
\usepackage{times}
\usepackage{graphicx}
\usepackage{latexsym}
\usepackage{amsmath}
\usepackage{amssymb}
\usepackage{mathtools}
\usepackage{tabularx}
\usepackage{mathdots}
\usepackage{mathrsfs}
\usepackage{algpseudocode}
\usepackage{algorithm}
\usepackage{verbatim}
\usepackage{tikz}
\usetikzlibrary{arrows,automata,chains,matrix,scopes}

\newcommand{\COH}{\mathsf{COH}}

\newcommand{\diam}{\mathsf{diam}}
\newcommand{\dist}{\mathsf{dist}}
\newcommand{\ecc}{\mathsf{ecc}}

\newcommand{\ein}[1]{s_{\mathsf{in}}(#1)}
\newcommand{\eout}[1]{s_{\mathsf{out}}(#1)}
\newcommand{\fin}[1]{f_{\mathsf{in}}(#1)}
\newcommand{\fout}[1]{f_{\mathsf{out}}(#1)}
\newcommand{\LM}{\mathsf{LM}}
\newcommand{\AP}{\mathsf{AP}}

\newcommand{\NDU}{\mathsf{NDU}}

\newcommand{\restrict}{\mathord{\upharpoonright}}
\newcommand{\SC}{\kappa}
\newcommand{\V}{\mathcal{W}}

\newcommand{\W}{\mathcal{W}}

\newcommand{\DO}{\mathsf{DO}}
\newcommand{\ZA}{\mathsf{ZA}}
\newcommand{\FT}{\mathsf{FT}}
\newcommand{\FB}{\mathsf{FB}}
\newcommand{\EN}{\mathsf{EN}}
\newcommand{\AS}{\mathsf{AS}}
\newcommand{\CM}{\mathsf{CM}}
\newcommand{\HE}{\mathsf{HE}}

\begin{document}

\title{Network, Popularity and Social Cohesion: A Game-Theoretic Approach}

\author{Jiamou Liu\inst{1} and  Ziheng Wei\inst{2}}
\institute{\ The University of Auckland, New Zealand\\
 \inst{1}\texttt{jiamou.liu@auckland.ac.nz} \ \inst{2}\texttt{z.wei@auckland.ac.nz}}

\maketitle

\begin{abstract}
In studies of social dynamics, cohesion refers to a group's tendency to stay in unity, which -- as argued in sociometry -- arises from the network topology of interpersonal ties between members of the group.
We follow this idea and propose a game-based model of cohesion that not only relies on the social network, but also reflects individuals' social needs. In particular, our model is a type of cooperative games where players may gain popularity by strategically forming groups. A group is socially cohesive if the grand coalition is core stable. We study social cohesion in some special types of graphs and
draw a link between social cohesion and the classical notion of structural cohesion \cite{WhiteHarary}. We then focus on the problem of deciding whether a given social network is socially cohesive and show that this problem is $\mathsf{CoNP}$-complete. Nevertheless, we give two efficient heuristics for coalition structures where players enjoy high popularity and experimentally evaluate their performances.
\end{abstract}

\section{Introduction}
Human has a natural desire to bind with others and needs to belong to groups. By understanding the basic instruments that generate coherent social groups, one can explain important phenomena such as the emergence of norms, group conformity, self-identity and social classes \cite{Festinger,Asch,HuismanBruggeman,Hogg}.  For example, studies reveal that on arrival to Western countries, immigrants tend to form cohesive groups among relatives and acquaintances in their ethnic communities,
which may hamper their acculturation into the new society \cite{NeeSanders}. Another study identifies cohesive groups of inhabitants in an Austrian village that correspond to stratified classes defined by succession to farmland ownership \cite{LilyanWhite}.

A social group arises when members of the groups are linked and develop bonds.
{\em Cohesion} refers to a tendency for a group to stay in unity, which is considered from two traditional  -- and seemingly opposing -- views: Firstly, group cohesion refers to a ``pulling force'' that draws members together \cite{Festinger}; Secondly, cohesion can also refer to a type of ``resistance'' of the group to disruption \cite{GrossMartin}. A common ground from both views is that cohesion amounts to a complex process characterised by both the micro-focus of psychology (fulfilment of personal objectives and needs), and the macro-focus of sociology (emergence of social classes) \cite{CarronBrawley}. A challenge is therefore to build a general but rigorous model to bridge the micro- and the macro-foci.


Most theories of group dynamics rely on two fundamental drives: {\em tasks}, and {\em social needs}.
Indeed, every group exists to accomplish a certain task; 
cooperation is desirable because  combining skills and resources leads to a better collective outcome. Based solely on this drive, cohesion becomes an issue of economics: how collective gains can be distributed among members to satisfy each member's goals. The theory of {\em cooperative games} tries to answer this question by assuming people as rational players who arrive at a stable outcome, i.e., a coalition formation where every coalition finds a stable division of the collective goods \cite{book}.
Social need is another important factor of group dynamics. A society embodies complex social relations such as friendship and trust. The theory of self categorization asserts that individuals mentally associate themselves into groups based on such traits \cite{Hogg}. Taking social relations into account, White and Harary describe cohesion as a network property and define {\em structural cohesion} \cite{WhiteHarary}; through this notion they prove that the two seemingly opposing views of cohesion (pulling force versus resistance to disruption) are in fact equivalent. Their work is then followed by intensive effort on community detection
 in the last 10-15 years \cite{Coscia,Fortunato}.

We identify insufficiencies in the existing mathematical models for social cohesion: 1) Cooperative game is a general framework on the economic process of resource allocation. While cohesion may imply stability,  cooperative games often do not capture cohesion, as they miss the crucial social network dimension. 2) Structural cohesion of a network refers to the minimum number of nodes whose removal results in network disintegration \cite{WhiteHarary}; this is a property of the network on the whole, and does not embody individual needs. Hence it fails to link the macro- with the micro-focus of group dynamics.

In this paper, we define cooperative games on social networks whose nodes are rational players. Outcomes of the game not only rely on the network of social relations, but also reflect individuals' social needs. Our model is consistent with the following theories: Firstly, we follow the network approach to study social phenomena, which is started by early pioneers such as Simmel and Dirkheim \cite{Durkheim,Wolff}. Secondly, our game-theoretic formulation is in line with theories in group dynamics that focus on the interdependence of group members \cite{Lewin}. Thirdly, we rigorously verify that networks with high structural cohesion also tend to be socially cohesive according to our definition.

To define payoffs, we adopt the following intuition: People prefer to be in a group where they are seen as valuable and influential members. Thus, the payoff of players in a sub-network should reflect in some sense {\em social positions}.
Social position is a multidimensional concept affected by a range of factors from behavioural and cognitive traits to structural and positional properties. Here we focus on the positional properties and follow the sociometric view that {\em popularity} -- an important indicator of social position -- arises from interpersonal ties such as liking or attraction among people \cite{LansuCillessen}; a person is popular if she is liked by a large portion of other people.  In particular, the authors of \cite{Conti} adopt the degrees of nodes, i.e., the numbers of edges attached to nodes, as a measure of popularity and identify economical benefits for a person to become popular. Hence in our games, we define the payoffs of players based on their degrees.

We now summarise the main contributions of the paper: (1) We propose {\em popularity games} on a social network and present a game-based notion of {\em social cohesion}, which refers to the situation when the grand coalition is core stable, a well-known stability concept. \  (2) To justify our model, we show consistency between popularity games and intuition over several special classes of networks. We also build a natural connection between  structural cohesion and our notion of social cohesion (Theorem~\ref{thm:structural cohesion}). \ (3) We prove that deciding whether a network is socially cohesive is computationally hard (Theorem~\ref{theorem:grand_co}) \
(4) Finally, we present two heuristics that decide social cohesion and compute group structures with high player payoffs and evaluate them by experiments.

\smallskip

\noindent {\bf Related works.}
 The series of works \cite{Narayanam,Michalak,Szczepanski,Szczepanski2,Szczepanski4} investigates game-based network centrality. Their aim is to capture a player's centrality using  various instances of semivalues, which are based on the player's expected payoff. In contrast, our study aims at games where the payoff of players are given a priori by degree centrality and focus on core stability.
\cite{Chen} uses non-cooperative games to explain community formation in a social network. Each player in their game decides among a fixed set of strategies (i.e. a given set coalitions); the payoff is defined based on {\em gain} and {\em loss} which depend on the local graph structures.  \cite{McSweeny} studies community formation through cooperative games. The payoffs of players are given by modularity and modularity-maximising partitions correspond to Nash equilibria. The focus is on community detection but not on social cohesion. Furthermore, our payoff function is not additively separable and hence does not extends from their model.
Lastly, our work is different from community detection \cite{Fortunato}. The notion of community structure originates from physics which focuses on a macro view of the network, while our work is motivated from group dynamics and focus on individual needs and preferences.

\paragraph*{\bf Paper organisation.} Section~\ref{sec:games} presents the game model and discusses notions of popularity and social cohesion. Section~\ref{sec:special_class} looks at several standard graph classes and characterizes core stability in each class. Section~\ref{sec:structural} links structural cohesion with our game-based social cohesion. Section~\ref{sec:complexity} proves that deciding social cohesion of a given network  is $\mathsf{CoNP}$-complete. Section~\ref{sec:experiment} relaxes social cohesion to a notion of social rationality and through experiments, this section connects this notion with community structures. Finally, Section~\ref{sec:conclusion} concludes and discusses future work.

\section{Popularity Games and Social Cohesion}\label{sec:games}

\paragraph*{\bf Network and games.}
A {\em social network} is an  unweighted graph $G=(V, E)$ where $V$ is a set of nodes and $E$ is a set of (undirected) edges. An edge $\{u,v\}\subseteq V$ (where $u\neq v$),  denoted by $uv$, represents certain social relation between $u,v$, such as attraction, interdependence and friendship. We do not allow loops of the form $uu$.
If $uv\in E$, we say that $u,v$ are {\em adjacent}.
A {\em path} from a node $u$ to a node $v$ is a finite sequence of nodes $u=u_1u_2 \ldots u_n=v$ where  $u_iu_{i+1}\in E$ for all $1 \leq i < n$. The network is {\em connected} if a path exists between any pairs of nodes. Let $G=(V,E)$ be a social network. We define a cooperative game on $G$ where each node in $V$ is a rational player.
The reader is referred to \cite{GameTheoryBook} for more details on cooperative game theory.

\begin{definition}
	A {\em cooperative game (with non-transferrable utility)} is a pair $G=(V, \rho)$ where $V$ is a set of players, and $\rho: V\times 2^V\rightarrow \mathbb{R}$ is a {\em payoff function}.
\end{definition}
\noindent A {\em coalition formation} of $G$ is a partition of $V$ $\W=\{V_1,\ldots,V_k\}$, i.e., $\bigcup_{1\leq i\leq k} V_i=V$, $\forall 1\mathord{\leq}i\mathord{<} j\mathord{\leq} k\colon V_i\cap V_j=\varnothing$; each set $V_i$ is called a {\em coalition}. The {\em grand coalition formation} is $\W_G=\{V\}$ where $V$ is called the {\em grand coalition}.
Cooperative games describe situations where players strategically build coalitions based on individual payoffs.
A predicted outcome of the game is a {\em stable} coalition formation in the sense that no set of players have the incentive to ``disrupt'' the formation by binding into a new coalition.
More precisely, let $(V, \rho)$ be a cooperative game.
Take a coalition formation $\W$ and
set $\rho_\W(u) \coloneqq \rho(u,S)$ where $u\in S$ and $S\in \W$.
\begin{definition}
	A non-empty set of players $H\subseteq V$ is {\em blocking for $\W$} if $\forall u\in H\colon \rho(u,H)>\rho_\W(u)$; In this case we say that $\W$ is {\em blocked} by $H$.
\end{definition}
In other words, if $S$ blocks $\W$, then every $u\in S$ will get a higher payoff if they join $S$.
\begin{definition}
	A coalition formation $\W$ of $G$ is {\em core stable} w.r.t. $(V,\rho)$ if it is not blocked by any set $H\subseteq V$.
\end{definition}

\paragraph*{\bf Popularity.}
{\em Social positions}, as argued in sociometric studies, arise from the network topology \cite{CillessenMayeux}. A long line of research studies how different {\em centralities} (e.g.  degree, closeness, betweenness, etc.) give rise to  ``positional advantage'' of individuals. In particular, degree centrality refers to the number of edges attached to a node. Despite its conceptual simplicity, degree centrality naturally represents (sociometric) {\em popularity}, which plays a crucial role in a person's self-efficacy and social needs \cite{Zhang,Conti}. Popularity depends on the underlying group: a person may be very popular in one group while being unknown to another. Hence individuals may gain popularity by forming groups strategically. We thus make our next definition.
The {\em sub-network induced} on  a set $S\subseteq V$ is $G\restrict S=(S, E\restrict S)$ where  $E\restrict S = E\cap S^2$. $\deg_S(u)$ denotes $|\{v\colon uv\in E\restrict S\}|$ and we write $\deg(u)$ for $\deg_V(u)$.
\begin{definition}
	The {\em popularity} of a node $u$ in a subset $S\subseteq V$ is $p_S(u)\coloneqq \deg_S(u)/|S|$.
\end{definition}
Note that $p_{\{u\}}(u)=0$ for every node $u$. If $u\in S$ has an edge to all other nodes in the graph $G\restrict S$, then $u$ is the most popular node in $S$ with $p_S(u)=(|S|-1)/|S|$. The popularity of any player is in the range $[0,1)$.
\begin{definition}
	The {\em popularity game} on $G=(V,E)$ is a cooperative game $\Gamma(G)=(V,\rho)$ where $\rho\colon V\times 2^V\to [0,1)$ is defined by $\rho(u,S)=p_S(u)$.
\end{definition}
An outcome of the popularity game $\Gamma(G)$ assigns any player $u$ with a coalition $S\hspace*{-0.5mm}\ni\hspace*{-0.5mm}u$. The sum of popularity of members of $S$ equals their {\em average degree} in $S$, i.e. $\sum_{u\in S}p_S(u)=\sum_{u\in S} \deg_S(u)/|S|=2|E\restrict S|/ |S|$. The average degree measures the {\em density} of the set $S$, which reflects the amount of interactions within $S$, and thus can be regarded as a collective utility. In this sense, the popularity game is {\em efficient} in distributing such collective utility among players according to their popularity.


\paragraph*{\bf Social cohesion.}
{\em Social cohesion} represents a group's tendency to remain united in satisfying members' social needs
 \cite{Cartwright}. We express cohesion through  {\em core stability} w.r.t. the popularity game $\Gamma(G)$:
Suppose a coalition formation $\W$ is not core stable. Then there is a set $S\subseteq V$ every member of which would gain a higher popularity in $S$ than in their own coalitions in $\W$. Thus there is a latent incentive among members of $S$ to disrupt $\W$ and form a new coalition $S$. On the contrary, a core stable $\W$ represents a state of the network that is resilience to such ``disruptions''.  Thus, when the grand coalition formation $\W_G=\{V\}$ is core stable, all members bind naturally and harmoniously into a single group and would remain so as long as the network topology does not change.

\begin{definition}
	A network $G=(V,E)$ is {\em socially cohesive} (or simply {\em cohesive}) if the grand coalition formation $\W_G$ is core stable w.r.t. the popularity game $\Gamma(G)$.
\end{definition}

\noindent{\bf Example 1.} Fig.~\ref{fig:example}(a) displays a network $G_1 = (V_1, E_1)$.
The popularity $p_V(i)$ is $1/3$ if $i=b,f$, and the popularity is $1/2$ if $i=a,c,d,e$. The set $\{a,b,c\}$ blocks $\W_{G_1}$ as each member has popularity  $2/3$. The only core stable formation is $\{\{a,b,c\},\{d,e,f\}\}$. Adding the edge $ad$ (shown in red) would make $G_1$ socially cohesive as the popularity of both $a,d$ in $V_1$ reaches $2/3$.
Fig.~\ref{fig:example}(b) displays another network $G_2=(V_2,E_2)$ where $p_V(a)=4/5$ and $p_V(i)=1/5$ for all $i=b,\ldots,e$. Note that this graph is socially cohesive as the grand coalition structure $\W_{G_2}$ is not blocked. However, adding the edge $bc$ (shown in red)  will destroy social cohesion as then $\{b,c\}$ blocks $\W_{G_2}$.
Social psychological studies often presume that more ties leads to higher cohesion; this example displays a more complicated picture: Adding an edge  may establish cohesion, but may also sabotage cohesion.

\begin{figure}
	\centering
	\resizebox{!}{2.5cm}{\includegraphics{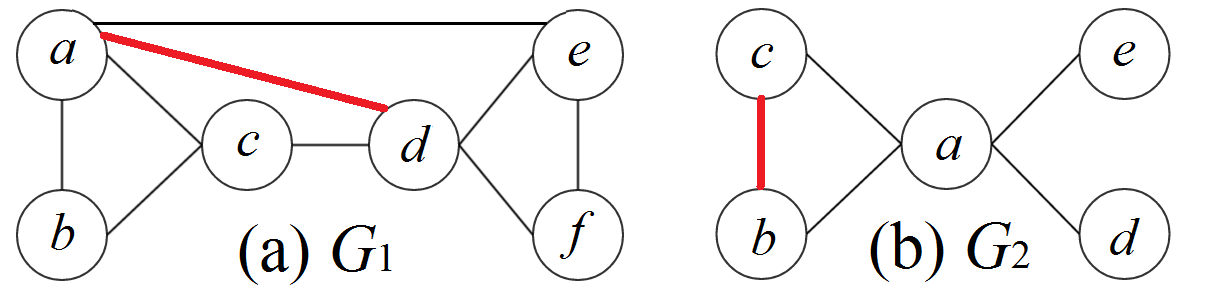}}\caption{  The graphs considered in Example 1 are in black. The added edges are highlighted in red.\label{fig:example}}
\end{figure}

\begin{theorem}[Connectivity]\label{thm:connected}
	If a coalition formation $\W$ of $G$ is core stable then any $S\in\W$ either consists of a set of isolated nodes, or induces a connected subgraph.
\end{theorem}
\begin{proof}
	Suppose that $S$ is not a set of isolated nodes, and that $S=V_1\cup V_2$ where $V_1,V_2$ are non-empty and  no edge exists between any pair in $V_1\times V_2$. Take a node $u\in S$ with non-zero degree, and say, without loss of generality, that $u\in V_1$. Then $$p_{V_1}(u)=\frac{\deg_{V_1}(u)}{|V_1|}> \frac{\deg_{V_1}(u)}{|S|}=\frac{\deg_{S}(u)}{|S|}=p_S(u).$$ Hence $\W$ does not contain $S$.
\qed \end{proof}

\noindent Theorem~\ref{thm:connected}
states that any two nodes (that are not isolated nodes themselves) not connected by a path have no incentive to be in the same coalition. Hence it is sufficient to only consider
coalitions that induce connected sub-networks of a social network.

\begin{definition}A set $S\subseteq V$ is called a {\em social group} of $G$ if $S$ induces a connected sub-network.
	A {\em group structure} is a coalition formation containing only social groups.
\end{definition}

\noindent The next theorem shows that social cohesion is inherently a {\em small group property}, i.e., socially cohesive networks have bounded size. 

\begin{theorem}\label{thm:small}
	Let $\delta(G)$ be the maximum degree of nodes in $G = (V,E)$. Then $G$ is socially cohesive only when $|V|\leq 2\delta(G)$ unless $|V|=1$.
\end{theorem}

\begin{proof}
	Suppose $|V|>2\delta(G)$ and $|V|>1$. If $E=\varnothing$, $G$ is not socially cohesive by Theorem~\ref{thm:connected}. Otherwise, pick an edge $uv$. Then $\max\{\deg(u),\deg(v)\}\leq \delta(G)<|V|/2$. This means that $\max\{p_V(u),p_V(v)\}<1/2$, and the edge $\{u,v\}$ forms a blocking set. Thus $G$ is not socially cohesive.
\qed \end{proof}

\section{Social Cohesion in Special Classes of Graphs}\label{sec:special_class}
We now investigate our games on some standard classes of graphs and characterize core stable group structures.

\smallskip

\paragraph*{\bf Complete networks.} A graph $G=(V, E)$ is {\em complete} if an edge exists between any pair of nodes. It corresponds to the tightest social structure where all members mutually interact.
Naturally, one would expect a complete network to be socially cohesive.

\begin{theorem}\label{thm:complete}
	Let $G=(V,E)$ be a complete network. The grand coalition is the only core stable group structure. Hence $G$ is socially cohesive.
\end{theorem}
\begin{proof} Any induced sub-network $G\restrict S$ of a complete network $G=(V,E)$ (where $S\subseteq V$) is also complete. Thus $$p_S(u)=\frac{|S|-1}{|S|}<\frac{|V|-1}{|V|}=p_V(u).$$
	Therefore any player's popularity is maximised in the grand coalition $V$.
\qed \end{proof}

\paragraph*{\bf Star networks.} A {\em star network} contains a node $c$ ({\em centre}), a number of other nodes $u_1,\ldots,u_m$ ({\em tails}) where $m>1$, and edges $\{cu_1,\ldots,cu_m\}$. Intuitively, the centre $c$ would like to be in a social group with as many players as possible, while a tail would like to be with as few others as possible.

\begin{theorem}\label{thm:star}
	A group structure $\W$ of a star network is core stable if and only if the centre is in the same social group with at least half of the tails. Thus, any star network is socially cohesive.
\end{theorem}

\begin{proof} Take any group structure $\W$ and suppose the centre $c$ is in a social group $S$ with  $\ell$ tails. Then $p_S(c)=\ell/(\ell+1)$ and for any tail $u_i\in S$, $p_S(u_i)=1/(\ell+1)$. All players not in $S$ has popularity 0 as their social groups contain only one node.
	\begin{itemize}
		\item
Suppose $\ell\geq m/2$. Take any set $S'\neq S$ that contains $c$. If $|S'|\leq |S|$, then $p_{S'}(c)\leq p_{S}(c)$. If $|S'|>|S|$, then  $p_{S'}(v)<p_{S}(v)$ for some tail $v$. In either case $S'$ does not block $\W$.
Hence $\W$ is core stable.

\item Suppose $\ell<m/2$. Then let $N$ be the set of all tails not in $S$. Then $|N|>\ell$. Hence the set $\{c\}\cup N$ is a blocking set for $\W$ as $p_{\{c\}\cup N}(c)>p_S(c)$.
\end{itemize}
	Thus $\W$ is core stable if and only if $\ell\geq m/2$.
\qed \end{proof}

\paragraph*{\bf Complete Bipartite Graph.}
A {\em complete bipartite graph} $K_{n,n}$ consists of disjoint sets of nodes $V_1,V_2$ with $n$ nodes each and $E=\{uv\mid u\in V_1, v\in V_2\}$ (where $m,n\neq 0$).
Let $\W$ be a group structure. For every $S \in \W$, we use $\ell(S)$ and $r(S)$ to denote $|\{v\mid v\in S \cap V_1 \}|$ and $|\{v\mid v\in S\cap V_2\}|$, respectively.

\begin{lemma}\label{lemma:knn}
	$\V$ is core stable only if $\forall S\in \V\colon\ell(S)\geq r(S)$.
\end{lemma}

\begin{proof}
	Suppose there is $S\in \V$ with $\ell(S)<r(S)$. Since $m\geq n$, there is $H\in \V$ with $\ell(H)>r(H)$. Take any  $u\in S\cap V_2$ and $v\in H\cap V_1$. Then we have
	$$p_S(u) = \frac{r(S)}{\ell(S)+r(S)}<\frac{1}{2}$$, \ \ and \ \ $$p_H(v)=\frac{\ell(H)}{\ell(H)+r(H)} <\frac{1}{2}.$$
	 Hence, the set $\{u,v\}$ blocks $\V$ as $p_{\{u,v\}}(u)=p_{\{u,v\}}(v)=1/2$.
\qed \end{proof}

\noindent We next characterize core stable group structure in $K_{n,n}$. In particular, {\em perfect matchings}, i.e., situations when every  $v\in V_1$ is matched with a unique player in $V_2$, are core stable.

\begin{theorem}
	A group structure $\V$ of $K_{n,n}$ is core stable if and only if $\forall S\in \V\colon \ell(S)=r(S)$.
\end{theorem}
\begin{proof}
	By Lemma~\ref{lemma:knn}, if $\V$ is core stable then $\forall S\in \V\colon\ell(S)=r(S)$. Conversely, if $\forall S\in \V\colon \ell(S)=r(S)$, then any $v$ has payoff $1/2$. Thus $\V$ is core stable as every $H\subseteq V$ contains some player with payoff at most $1/2$.
\qed \end{proof}

 We now turn our attention to $K_{m,n}$ with arbitrary $m\geq n\geq 1$ and focus on a special type of group structures: A {\em clan structure} $\V$ is a group structure that contains at most one non-singleton social group, called the {\em clan}; all other social groups contain only single players, called the {\em exiles}. The number $\iota(\V)$ is the number of exiles, i.e., $\iota(\V)=|\{S\in \V \mid  |S|=1\}|$. It is clear that any group structure of a star network is a clan structure.
Theorem~\ref{thm:star} then becomes a special case (when $n=1$) of the next theorem, which characterizes core stable clans structures of $K_{m,n}$.

\begin{theorem}\label{thm:inequitable_coalition}For any $m \geq n > 0$, a clan structure $\V$ of $K_{m,n}$  is core stable if and only if  the clan $S$ contains all nodes in $V_2$ and  $\ell(S) \geq  \max\{n,\iota(\V)\mathord{\cdot}n\}$. Thus, $K_{m,n}$ is socially cohesive.
\end{theorem}
\begin{proof}
	Suppose $\V$ is a core stable clan structure. Lemma~\ref{lemma:knn} implies that the clan $S$ contains all nodes in $V_2$ and $\ell(S)\mathord{\geq}n$. If $\ell(S)\mathord{<}\iota(\V)\mathord{\cdot}n$, $(\iota(\V)\mathord{+}1)\ell(S)\mathord{<}\iota(\V)(\ell(S)\mathord{+}n)$ and thus $\frac{\ell(S)}{\ell(S)+n}<\frac{\iota(\V)}{\iota(\V)+1}$. Then the set $\{v\}\cup X$ blocks $\V$ where $v\in S\cap V_2$ and $X$ is the set of exiles in $G$. Hence we must have $\ell(S)\geq \max\{n,\iota(\V)\cdot n\}$.
	
	Conversely, suppose $V_2\subseteq S$ and $\ell(S)\geq \max\{n,\iota(\W)n\}$. Assume for a contradiction that $H$ blocks $\V$. Then for any $v\in H\cap V_2$, $p_H(v)>p_S(v)=\frac{\ell(S)}{\ell(S)+n}\geq \frac{\iota(\W)}{\iota(\W)+1}$. This means that $H$ must contain a player $u\in V_1$ that belongs to $S$, and we must have $p_H(u)>p_S(u)=\frac{n}{\ell(S)+n}$. However,
	\begin{align*}
	1&=\frac{n}{\ell(S)+n}+\frac{\ell(S)}{\ell(S)+n}\\
	 &=p_S(u)+p_S(v)\\
	 &<p_H(u)+p_H(v)\\
	 &\ \ \ =\frac{\ell(H)}{\ell(H)+r(H)}+\frac{r(H)}{\ell(H)+r(H)}=1
	\end{align*}
	Contradiction. Thus $\V$ is core stable.
\qed \end{proof}

\section{Structural Cohesion and Social Cohesion}\label{sec:structural}
White and Harary in \cite{WhiteHarary} describe group cohesion using {\em graph connectivity}.

\begin{definition}[White and Harary \cite{WhiteHarary}]\label{eqn:critical} Let $G$ be a connected graph. The {\em structural cohesion} $\SC$ of a connected graph $G$ is the minimal number of nodes upon removal of which $G$ become disconnected.
\end{definition}

\noindent As stated in \cite{WhiteHarary}, a larger  $\SC$  implies that  $G$ is more resilient to conflicts or the departure of group members, and is thus more cohesive. Moreover, Menger's theorem states that $\SC$ is the greatest lower bound on the number of paths between any pairs of nodes. Hence $\SC$ is a reasonable measure of cohesion. We next link $\SC$ with our notion of social cohesion.
In \cite{Granovetter}, a pair $uv\notin E$ is seen as a type of ``structural hole'' that forbids communication and is thus referred to as an {\em absent tie}. 
For each $S\subseteq V$ and $u\in S$ we define the following:
\begin{enumerate}
\item $\fin{u,S}\coloneqq \deg_S(u)$ is the number of actual ties of $u$ within the group $S$,
\item $\fout{u,S}\coloneqq |\{v\notin S\mid uv\in E\}|$ is the number of actual ties of $u$  outside $S$,
\item $\ein{u,S}\coloneqq |S|-\fin{u,S}$  is the number of absent ties (including $u$ itself) in $S$, and
 \item $\eout{u,S}\coloneqq |\{v\notin S\mid uv\notin E\}|$ is the number of absent ties outside $S$.
\end{enumerate}
These variables give rise to a characterization of social cohesion.
Intuitively, if $S\subseteq V$ is a blocking set, each member $u$ tends to have many actual ties within $S$ and  absent ties outside $S$, i.e., high $\fin{u,S}$ and $\eout{u,S}$, and $u$ tends to have few absent ties in $S$ and actual ties outside $S$, i.e., low $\fout{u,S}$ and $\ein{u,S}$. Thus, we define for all $S\subseteq V$, $u\in S$,
\begin{equation}
\gamma(u,S) \coloneqq \fin{u,S}\eout{u,S}-\fout{u,S}\ein{u,S}
\end{equation}

\begin{lemma}\label{lmm:crit}
	For all $S\subseteq V$, $S$ blocks $\W_G=\{V\}$ if and only if $\forall u\in S\colon \gamma(u,S)>0$.
\end{lemma}

\begin{proof}
	For each $u\in S$, $p_S(u)= \frac{\fin{u, S}}{\fin{u, S}+\ein{u, S}}$ and
	\begin{align*}
		p_V(u) =\frac{\deg(u)}{|V|}
		=\frac{\fin{u, S}+\fout{u, S}}{\fin{u, S}+\fout{u, S}+\ein{u, S}+\eout{u, S}}
	\end{align*}
	The set $S$ blocks $\W_G$ if and only if $\forall u\in S\colon p_S(u) >p_V(u)$, which can be shown to be equivalent to $\fin{u, S} \cdot \eout{u, S} > \fout{u, S} \cdot \ein{u,S}$ using the above equalities.
\qed \end{proof}

\noindent A network $G$ contains a {\em minimal cut} $A_0\subseteq V$ of size $\SC$, i.e., removing $A_0$ from $G$ decomposes the graph into $m$ distinct connected components $A_1,\ldots,A_m\subseteq V$ where $m\geq 2$, i.e., $G$ contains disjoint sets $A_1,\ldots,A_m\subseteq V$ such that each $G\restrict A_i$ is connected, and for any $j\neq i$ all paths between $G\restrict A_i$ and $G\restrict A_j$ go through $A_0$.
 We further assume  that $|A_1|\leq\cdots \leq |A_m|$ and $A_0$ is chosen in a way where $|A_1|$ is as small as possible.
Let $\chi$ be the size $|A_1|$, and let $\mu$ be the largest possible length $m$ of the sequence of $A_i$'s. 
 We first look at the case when $\SC=1$.

 \begin{lemma}\label{lem:sc1}
 	If $\SC = 1$ and $G$ is socially cohesive, then $\chi < 2$.
 \end{lemma}
 \begin{proof}Suppose $\SC=1$ and $\chi\geq 2$.
 	Let $(A_1,\ldots, A_m)$ be an optimal cut sequence. Take $u\in A_1$. As  $G$ contains a cut node, $\fout{u,A_1}\leq 1$ and $\eout{u,A_1}\geq |V|-|A_1|-1=
 	|V|-\chi-1$. Then $\gamma(u,A_1)\geq \fin{u,A_1}\cdot (|V|-\chi-1)-\ein{u,A_1}$.
 	Since $\fin{u,A_1}+\ein{u,A_1}=\chi$,
 	\begin{align*}
 	\gamma(u,A_1)&\geq \fin{u,A_1}(|V|-\chi-1)-(\chi-\fin{u,A_1})\\
 	&=\fin{u,A_1}(|V|-\chi)-\chi.
 	\end{align*} Since $|V|-\chi\mathord{>}\chi$, $\gamma(u,A_1)\mathord{>}0$.
 	Thus by Lemma~\ref{lmm:crit}, $A_1$ blocks the grand coalition $\W_G$.
 \qed \end{proof}

 \noindent We now generalize Lemma~\ref{lem:sc1} to graphs with higher structural cohesion.

 \begin{lemma}\label{lem:mu2}
 	Suppose $\mu>2$. Then any network $G$ is socially cohesive only if $\chi< \frac{\SC}{\mu-2}$.
 \end{lemma}

 \begin{proof}Suppose $\mu>2$.
 	Take an optimal cut sequence  $(A_1,...,A_\mu)$  and $u\in V_1$. Since $\deg(u)< \chi+\SC$ and $|V|\geq \mu\chi+\SC$, we have $p_V(u) < \frac{\chi+\SC}{\mu\cdot \chi+\SC}$. Suppose $\chi\geq\frac{\SC}{\mu-2}$. Then $ \mu\chi-2\chi \geq\SC$. One can then derive that $p_V(u)<\frac{\chi+\SC}{\mu\cdot \chi+\SC}\leq\frac{1}{2}$. Thus any edge $\{u,v\}$ in $G\restrict V_1$ forms a blocking set of the grand coalition formation $\W_G$.
 \qed \end{proof}

 Lemma~\ref{lem:mu2} can be used as a (semi-)test for social cohesion when $\mu>2$:  whenever $\chi$ exceeds $\frac{\SC}{\mu-2}$, $G$ is not socially cohesive. Clearly, more graphs become socially cohesive as $\SC$ gets larger.
 Summarizing Lemma~\ref{lem:sc1} and Lemma~\ref{lem:mu2}, we obtain the following result.

 \begin{theorem}\label{thm:structural cohesion}
 	Let $G$ be a network.
 \begin{itemize}
  \item If $\SC=1$, then $G$ is not socially cohesive as long as $\chi\geq 2$.
  \item If $\SC>1$ and $\mu>2$, then $G$ is not socially cohesive as long as $\chi\geq\frac{\SC}{\mu-2}$
  \end{itemize}
 \end{theorem}

\begin{remark} The only case left unexplained is when $\SC>1$ and $\mu=2$. In this case there exist graphs with arbitrarily large $\chi$ but are socially cohesive.
\end{remark}
\section{The Computational Complexity of Deciding Social Cohesion} \label{sec:complexity}
We focus on the computational problem of deciding if a network is socially cohesive. More precisely,
We are interested in the decision problem $\COH$: 
\begin{description}
	\item[INPUT] A network $G=(V,E)$
	\item[OUTPUT] Decide if $G$ is socially cohesive.
\end{description}
Instead of considering networks in general, we restrict our attention to a special type of networks.
The {\em distance} between two nodes $u$ and $v$, denoted by $\dist(u,v)$, is the length of a shortest path from $u$ to $v$ in $G$. The {\em eccentricity} of $u$ is $\ecc(u)=\max_{v\in V} \dist(u,v)$. The {\em diameter} of the network $G$ is $ \diam(G)=\max_{u\in V} \ecc(u)$.
\begin{definition}
	A graph $G=(V, E)$ is {\em diametrically uniform} if all $v\in V$ have the same eccentricity; otherwise $G$ is called {\em non-diametrically uniform}.
\end{definition}
 We use $\NDU_2$ to denote the set of all  non-diametrically uniform connected graphs whose diameter is at most $2$. Our goal is to show that the $\COH$ problem is computationally hard already on the class $\NDU_2$. The following is a characterization theorem for $\NDU_2$ graphs.

\begin{theorem}\label{thm:eccentricity}
	The network $G$ belongs to $\NDU_2$ if and only if its nodes $V$ can be partitioned into two non-empty set $V_1$ and $V_2$, where
	$V_1=\{u\mid vu\in E\text{ for all } v\neq u\}$.
\end{theorem}
\begin{proof}
	Since $G$ has diameter 2 and is not diametrically uniform, there is a non-empty set $V_1$ of nodes with eccentricity 1, and the other nodes (call them $V_2$) have eccentricity 2.  The sets $V_1,V_2$ satisfy the condition in the theorem.
\qed \end{proof}

\noindent Let $G$ be a graph in $\NDU_2$. We call $\{V_1,V_2\}$ as described in Theorem~\ref{thm:eccentricity} the {\em eccentricity partition} of $G$. We first present some simple properties of $\NDU_2$.

\begin{lemma}\label{lmm:ndu_separatist}
	The network $G$ in $\NDU_2$ is socially cohesive if and only if no set $S\subseteq V_2$ blocks $\W_G$.
\end{lemma}
\begin{proof}
	One direction (left to right) is clear. Conversely, suppose the network is not socially cohesive. Let $S \subset V$ be a blocking set of the grand coalition formation, i.e., $\forall u\mathord{\in} S\colon \ p_S(u) > p_V(u)$. If $S\cap V_1 \neq \varnothing$. Then $\forall u\in S\cap V_1\colon p_V(u)= \frac{|V|-1}{|V|}$. However, $p_V(u)\geq  \frac{|S|-1}{|S|}\geq p_S(u)$ which contradicts that fact that $S$ is a blocking set.
\qed \end{proof}

By the lemma above, the structure of $G\restrict V_2$ is crucial in determining social cohesion of $G$. For any $S\subseteq V_2$ and $u\in S$, we recall the notions $\fin{u,S},\fout{u,S},\ein{u,S}$, and $\eout{u,S}$ from Section~\ref{sec:structural}, but re-interpret these values within the sub-network $G\restrict V_2$. Hence, we now set  $\fout{u,S}$ as $|\{v\in V_2\setminus S\mid uv\in E\}|$, i.e., the number of ties that $u$ has within $V_2$ but not in $S$, the other variables remain as originally defined. Thus
\begin{equation}\label{eqn:substitute}
|V_2|=\fin{u, S}+\fout{u, S}+\ein{u, S}+\eout{u, S}
\end{equation}
We then define the value
\[
\lambda(u,S)=\frac{\fin{u, S}\cdot \eout{u, S}}{\ein{u, S}} - \fout{u, S}
\]

\begin{theorem}\label{lmm:grand_ne_su_condition}
	A network	$G$  in $\NDU_2$ is socially cohesive if and only if for all $S\subseteq V_2$ there exists $v\in S$ such that $ |V_1| \geq \lambda(v, S)$
\end{theorem}
\begin{proof}
	By Lemma~\ref{lmm:ndu_separatist}, we only need to examine subsets $S\subseteq V_2$. Every $u\in S$ has $$p_S(u) = \frac{\fin{u, S}}{\fin{u, S}+\ein{u, S}}\text{ and } p_V(u) = \frac{\fin{u, S}+\fout{u, S}+|V_1|}{|V_2|+|V_1|}.$$
	A set $S\subseteq V_2$ blocks $\W_G=\{V\}$ if and only if $\forall u\in S\colon p_S(u)>p_V(u)$. Note that	$$\frac{\fin{u, S}}{\fin{u, S}+\ein{u, S}} > \frac{\fin{u, S}+\fout{u, S}+|V_1|}{|V_2|+|V_1|}\text{ if and only if}$$ $\fin{u,S}(|V_1|\mathord{+}|V_2|)>(\fin{u, S}\mathord{+}\ein{u, S})( \fin{u, S}\mathord{+}\fout{u, S}\mathord{+}|V_1|)$.

	Applying (\ref{eqn:substitute}), $S$ blocks $\W_G$ if and only if $\forall u\in S: p_S(u)>p_V(u)$, if and only if  $\forall u\in S:|V_1|<\frac{\fin{u, S}\eout{u, S}}{\ein{u, S}}-\fout{u, S}=\lambda(u,S)$, as required.
\qed \end{proof}

We now give a sufficient condition for social cohesion of an $\NDU_2$ network. The {\em size} of a network is its number of nodes. A {\em clique} is a complete subgraph. The {\em clique number} of $G$, denoted by  $\omega(G)$, is the size of the largest clique. Tur\'{a}n's theorem relates $\omega(G)$ with the number of edges in $G$:

\begin{theorem}[Tur\'an \cite{Turan}]\label{thm:turan} For any $p\geq 2$, if a graph $G$ with $n$ nodes has more than $\frac{p-2}{2(p-1)} n^2$ edges, then $\omega(G)\geq p$.
\end{theorem}

\begin{lemma}\label{turan}
	For any social group $S\subseteq V$, there exists $u\in S$ with $\frac{\fin{u, S}}{\ein{u, S}}\mathord{\leq}\omega(G\restrict S)\mathord{-}1$.
\end{lemma}

\begin{proof}
	Let $k=\omega(G\restrict S)$ and suppose for all $u\in S$, $\frac{\fin{u, S}}{\ein{u, S}}> k-1$. Since $|S|=\fin{u, S}+\ein{u, S}$,
	\begin{align*}
    k-1&<\fin{u, S}/(|S|-\fin{u, S})\\
	\frac{\fin{u, S}}{|S|-\fin{u, S}} &> k-1\\
	\fin{v,S} &> (c-1)(|S|-\fin{v, S}) \\
	\fin{u,S} &> k|S|-k\fin{u, S}-|S|+\fin{u, S}\\
	k\fin{u, S}&> |S|\cdot (k-1)\\
	\fin{u, S} &> (k-1)|S|/k
	\end{align*}
	Thus $E\restrict S$ contains  $>\frac{k-1}{2k}|S|^2$ edges. By Theorem~\ref{thm:turan}, $G\restrict S$ contains a size-$(k+1)$ clique, contradicting $k$'s definition.
\qed \end{proof}

The following is a sufficient condition for social cohesion of an $\NDU_2$ network $G$.
Intuitively, the set $V_1$ contains the most socially active members -- those who interact with everyone else. Hence
they serve as ``socializers'' who hold the group together. The larger $V_1$ gets, the more likely $G$ will be socially cohesive. There is a bound such that once $|V_1|$ exceeds it, the network $G$ is guaranteed to be socially cohesive.

\begin{lemma}\label{theorem:grand_co}
	Suppose $c=\omega(G\restrict V_2)$ and $|V_2|>c(c-1)$. Then $G$ is socially cohesive whenever $|V_1|\geq (c-1)(|V_2|-c)$.
\end{lemma}

\begin{proof}
	Suppose $|V_2|>c(c-1)$, $|V_1|\geq (c-1)(|V_2|-c)$. Take any $S\subseteq V_2$. If $S$ has a size-$c$ clique,
	by Lemma~\ref{turan} there exists $u \in V_2$ with $\frac{\fin{u, S}}{\ein{u, S}} \leq c-1$. Since $\eout{u,S} \leq |V_2|-c$, we have
	$$|V_1| \geq (c-1)(|V_2|-c) \geq \frac{\fin{u, S}}{\ein{u, S}}\eout{u, S} > \lambda(u, S)$$
	and thus  $G$ is socially cohesive by Theorem~\ref{lmm:grand_ne_su_condition}.
	
	If $S$ contains no clique of size $c$, then $\omega(G\restrict S)\mathord{\leq}c\mathord{-}1$. Let $k=\omega(G\restrict S)$. By Lemma~\ref{turan}, $\frac{\fin{u, S}}{\ein{u, S}}\mathord{\leq}k\mathord{-}1\mathord{<}c\mathord{-}1$. Thus
	$c\mathord{-}\frac{|V_2|}{\ein{v, S}}\mathord{\geq}1$.
	Since $\eout{u, S}\mathord{=}|V_2|\mathord{-}\fin{u, V_2}\mathord{-}\ein{u, S}$,
	\begin{align*}
	\lambda(u, S)&= \frac{\fin{u, S}}{\ein{u, S}}(|V_2|-\fin{u, V_2}-\ein{u, S}) - \fout{u, S}\\
	& =\frac{\fin{u, S}}{\ein{u, S}}(|V_2|-\fin{u, V_2}) - \fin{u, S}-\fout{u, S}\\
	&=\frac{\fin{u, S}}{\ein{u, S}}(|V_2|-\fin{u, V_2})-\fin{u, V_2}.
	\end{align*}
	Furthermore, since $\fin{u, S}=|S|-\ein{u,S}$,
	\begin{align*}
	\lambda(u, S)& = \frac{(|S|-\ein{u,S}) (|V_2|-\fin{u, V_2})}{\ein{u, S}} - \fin{u, V_2} \\
	&=\frac{|S|\cdot|V_2|-|S| \fin{u, V_2}-|V_2| \ein{u, S}}{\ein{u, S}}\\
	&=\frac{|S|\cdot |V_2|}{\ein{u, S}} -\frac{|S| \fin{u, V_2}}{\ein{u, S}}-|V_2|.
	\end{align*}
	Hence, $|V_1|-\lambda(u,S)$ is at least
	\begin{align*}
	& (c-1)(|V_2|-c) - \left(\frac{|S|\cdot |V_2|}{\ein{u, S}} -\frac{|S|\fin{u, V_2}}{\ein{u, S}}-|V_2|\right)\\
	& = |V_2|\left(c-\frac{|S|}{\ein{u, S}}\right)-c(c-1)+\frac{|S|\fin{u, V_2}}{\ein{u, S}} \geq 0
	\end{align*}
	The last step is by $c\mathord{-}\frac{|S|}{\ein{v, S}}\mathord{\geq}1$ and $|V_2|\mathord{>}c(c\mathord{-}1)$.
	By Theorem~\ref{lmm:grand_ne_su_condition}, $G$ is socially cohesive.	
\qed \end{proof}

\begin{theorem}\label{thm:complexity}The problem $\COH$ is $\mathsf{CoNP}$-complete. Furthermore $\mathsf{CoNP}$-hardness holds already for the class $\NDU_2$.
\end{theorem}
\begin{proof} The complement of $\COH$, $\overline{\COH}$, asks whether a set $S$ blocks the grand coalition $\W_G$ of a given network $G$; this problem is clearly in $\mathsf{NP}$ and thus $\COH$ is in $\mathsf{CoNP}$. For hardness, we reduce $\mathsf{MaxClique}$ (asking whether a graph contains a clique of a given size $k$) to $\overline{\COH}$. $\mathsf{CoNP}$-hardness of $\COH$ then follows from the $\mathsf{NP}$-hardness of $\mathsf{MaxClique}$ \cite{GareyJohnson}.

	\begin{algorithm}[!]
		\caption{Construction of  $H$ given $G = (V,E)$ and $k > 2$}\label{construction}
		\begin{algorithmic}[1]
			\State Set $d \coloneqq k\cdot \max\{\deg(u) \mid u\in V\}$
			\State Create $G'$ by adding $k(k\mathord{-}1)\mathord{+}d$ isolated nodes (and no new edges) to $G$
			\State Let $V_2$ be the set of nodes in $G'$ (which contains all nodes in $V$ and new nodes)
			\State Create a complete graph with $(k\mathord{-}1)(|V_2|\mathord{-}k)\mathord{-}d$ nodes; Let $V_1$ be the set of these nodes
			\State Create edges $\{uv\mid u\in V_1,v\in V_2\}$ to connect $V_1,V_2$. The resulting graph is $H$.
		\end{algorithmic}
	\end{algorithm}

\noindent	To this end, we construct, for a given $G=(V,E)$ and $k>2$, a graph $H\in \NDU_2$ as in Alg.~\ref{construction}. Our goal is to show that $H$ is not socially cohesive if and only if $G$ contains a clique of size $k$. It is clear that $H$ is a $\NDU_2$ network with eccentricity partition $\{V_1,V_2\}$. Let $c=\omega(G)$. Suppose $c<k$. By definition of $V_1$ and $V_2$, we have $$|V_1|-(c-1)(|V_2|-c)=|V_2|(k-c)-k(k-1)+c(c-1)  - d.$$
Since $k\mathord{>}c$ and $|V_2|\mathord{>}k(k\mathord{-}1)\mathord{+}d$, $|V_1|\mathord{\geq}(c\mathord{-}1)(|V_2|\mathord{-}c)$. By Theorem~\ref{theorem:grand_co}, $H$ is cohesive.

Conversely, suppose $\omega(G)\geq k$. Let $C\subseteq V_2$ be a clique of size $k$. Take $u\in C$.
Since $\fin{u,C} = k - 1$, $\eout{u, C} = |V_2| - k - \fout{u, C}$ and $\ein{u, C}=1$,
\begin{align*}
\lambda(u, C) &= (k-1)(|V_2|-k-\fout{u, C}) - \fout{u, C}\\
&= (k-1)(|V_2|-k)-k\cdot \fout{u, C}
\end{align*}
Hence, $\lambda(u,C)\mathord{-}|V_1|=d-k\fout{u, C}$.
Since $d\geq k \deg(u)$, $\lambda(u, C)-|V_1|>0$. By Lemma~\ref{lmm:grand_ne_su_condition}, $G$ is not socially cohesive. Therefore, $G$ contains a clique of size $k$ if and only if $H$ is not socially cohesive and the reduction is complete.
\qed \end{proof}

\section{Efficient Heuristics} \label{sec:experiment}
We propose two heuristics that construct group structures of a given network where players enjoy high popularity. These heuristics (partially) solve $\COH$ despite $\COH$'s inherent complexity: Each heuristic builds a group structure $\W$ and checks if any set $S \in \W$ blocks  $\W_G$.  If $G$ is socially cohesive, then no such $S$ will be found; On the other hand, if a blocking set $S$ is found, $G$ is surely not socially cohesive. 


\smallskip

\noindent {\bf Heuristic 1: Louvain's method ($\LM$)} We observe that blocking sets of $W_G$ are usually tightly connected within, but are sparsely connected with nodes outside. This property corresponds to the well-studied notion of {\em communities}  \cite{Fortunato}. Therefore, the first heuristic uses a well-known  community detection algorithm, Louvain's method \cite{Blondel}, to compute a group structure in $G$.


\smallskip

\noindent {\bf Heuristic 2: Average payoff ($\AP$)} The second heuristic aims to optimize the average payoffs of members of a coalition. Socially cohesive networks usually have small diameters ($\leq 2$). Thus we consider neighborhood $N(v)\coloneqq\{v\}\cup \{u\mid uv\in E \}$ of players $v \in V$. In Alg.~\ref{algo:ap}, let $\nu(S)$ be the average payoff $\sum_{u\in S}\rho(u, S)/|S|$ in any set $S\subseteq V$.


			
	\begin{algorithm}
		\caption{$\AP$: Given a network $G=(V, E)$}\label{algo:ap}
		\begin{algorithmic}[1]
            \State Initialize set $V'\coloneqq V$ and a group structure $\W\coloneqq \varnothing$
			\While{$|V'|>0$}
            \State compute $\nu_w\coloneqq \nu(N(w)\cap V')$ for every $w\in V'$
			\State $S\coloneqq N(v)\cap V'$ such that $\forall u\in V'\colon \nu_u\leq \nu_v$
            \State $\W\coloneqq \W\cup \{S\}$ \hfill $\rhd$ Add $S$ to $\W$
			\State $V'\coloneqq V'\setminus S$
			\EndWhile
			\State \Return $\W$
		\end{algorithmic}
	\end{algorithm}

\noindent {\bf Experiments.} We evaluate the heuristics on graphs of size $n=5,\ldots,18$. For each size, let $s$ be the total number of graphs, $b$ be the number of graphs for which the heuristic finds a blocking set, and $c$ be the number of socially cohesive graphs. The heuristic thus correctly solves $\COH$ for $b+c$ graphs. Hence  the heuristic has {\em accuracy} $(b+c)/s$. For $n=5,6,7$, we exhaustively enumerate all connected graphs; $\LM$ has accuracy $96.8\%$, $97.7\%$ and $83.1\%$, resp., while $\AP$ has $85.9\%$, $68.5\%$ and $69.2\%$, resp. For each $n=8,\ldots,18$, we uniformly randomly sample $10^5$ graphs of size $n$. As shown in Fig.~\ref{fig:enumerated_accuracy}(a), both heuristics achieve high accuracy. As $n$ increases, socially cohesive graphs become increasingly rare. The results show that the heuristics successfully find blocking sets in almost all cases, with $\AP$ performs slightly better ($100\%$ accuracy for $n\geq 12$). We then consider the coalitions constructed by the heuristics. Fig.~\ref{fig:enumerated_accuracy}(b) shows that while $\LM$ fails to output core stable group structures, $\AP$ achieves core stability in $60\%$ of the sampled cases when $n\geq 7$. Nevertheless, Fig.~\ref{fig:enumerated_accuracy}(c) shows that, compared to the payoffs in the grand coalition, more nodes get a higher payoff in the coalitions identified by $\LM$ than in the coalitions identified by $\AP$.


\begin{figure}
	\centering
	\includegraphics[width=\textwidth]{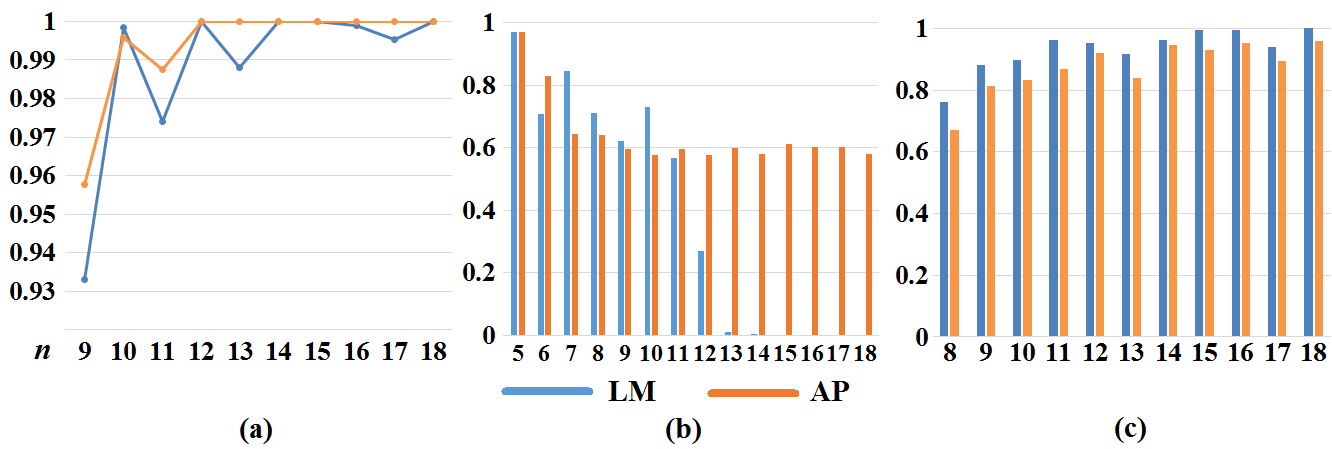}
	\caption{  The results of running heuristics $\LM$ and $\AP$ on graphs with $n=5,\ldots,18$ nodes. (a) Accuracy (b) The percentage of graphs for which the heuristics find a core stable group structure (c) The percentage of nodes with an increased payoff in the found group structure than in the grand coalition.}
	\label{fig:enumerated_accuracy}
\end{figure}

\smallskip

\noindent {\bf  Real world networks.}
We further evaluate the heuristics on 8 real-world networks: karate club $\ZA$ \cite{Zachary},  dolphins $\DO$ \cite{Dolphin},   college football $\FT$ \cite{GirvanNewman}, Facebook $\FB$, Enron email network $\EN$ \cite{Leskovec}, and three physics collaboration networks $\AS$, $\CM$ and $\HE$ \cite{Leskovec2}. We only use the largest components in each network; see details in Table.~\ref{table:real_world_1}. Expectedly, none of these networks are socially cohesive.
The box-and-whisker diagrams in Fig.~\ref{fig:real-box} show the distribution of payoffs of players in the grand coalition as well as in the coalitions output by each heuristic (outliers omitted). In all cases, the heuristics improve players' payoffs considerably compared to the grand coalition, while $\AP$ in particular achieves higher payoffs. Furthermore, Fig.~\ref{fig:increasing-rate} shows that all nodes get higher payoffs through $\LM$.
In summary, both of the heuristics are useful in computing coalitions; while $\LM$ may benefit a larger portion of players, $\AP$ tends to obtain higher payoffs. 

\begin{table}
	\centering
	\caption{Statistics of the real-world networks used in our experiment. The parameters $N$, $E$, $C$ denote the number of nodes, edges, and communities, respectively.}
	\label{table:real_world_1}
	\begin{tabularx}{0.5\textwidth}{@{\extracolsep{\fill}}|c|c|c|c|c|}
		\hline
		& {$\ZA$}  &{$\DO$}  &{$\FT$} &{$\FB$}      \\ \hline
		$N$ &34  & 62  &  115  &3959     \\ \hline
		$E$ &78  &159  &613  &84243    \\ \hline
		$C$ &4  &5  &10  &27      \\ \hline\hline
	    &{$\EN$} &{$\AS$}  &{$\CM$}  &{$\HE$}    \\ \hline
	    $N$   &36692 &18772  &23133  &12008    \\ \hline
		$E$   &183831 &198110  &93497  &118521    \\ \hline
		$C$  &248 &36  &58  &37     \\ \hline
	\end{tabularx}	
\end{table}


\begin{figure}
	\centering
	\includegraphics[width=\textwidth]{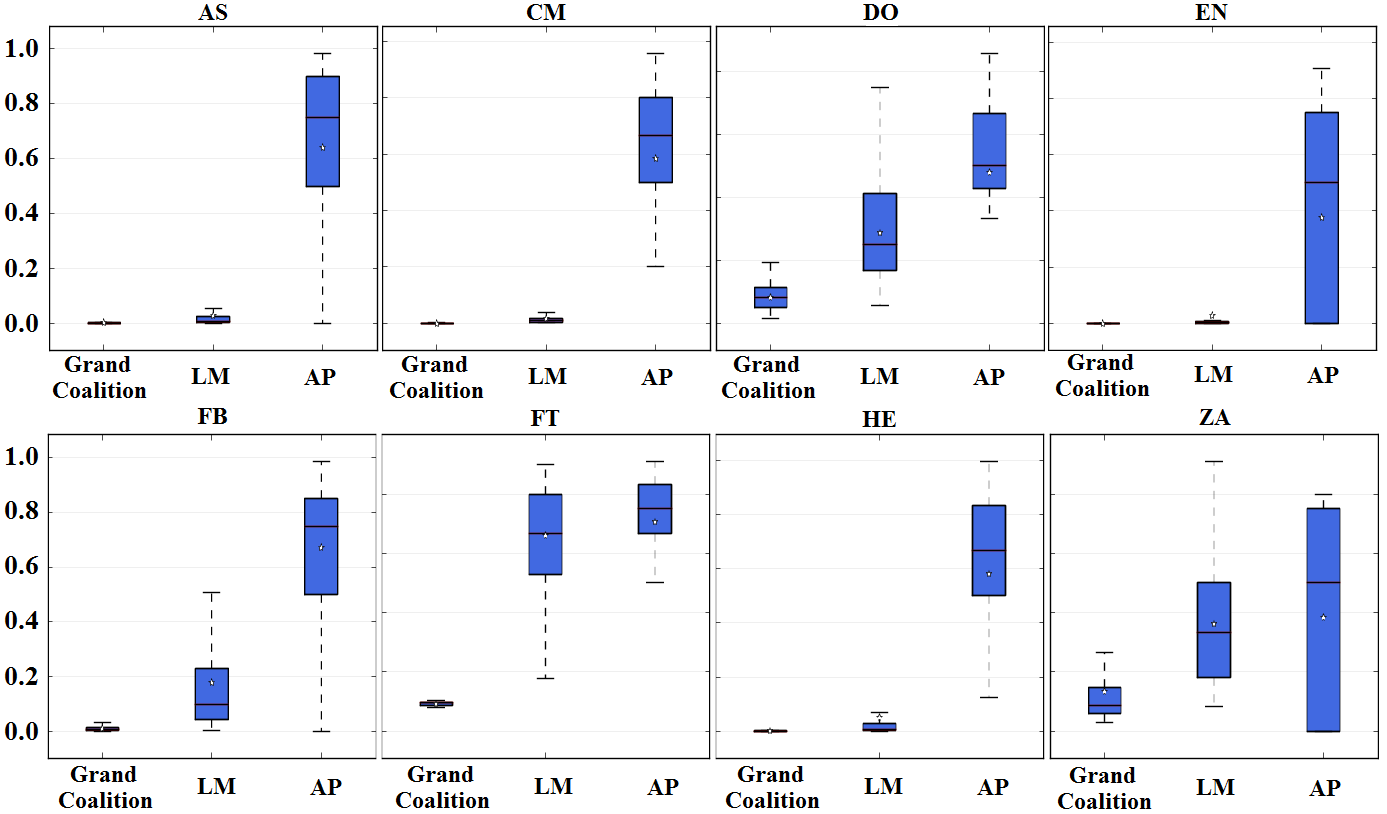}
	\caption{  The distribution of payoffs of players in the grand coalition and in the coalitions computed by the heuristics in the real-world networks.}
	\label{fig:real-box}
\end{figure}

\section{Conclusion and future work}\label{sec:conclusion}
This paper investigates social cohesion through a type of cooperative games. We show that social cohesion are closely related to structural cohesion and demonstrate that checking social cohesion is computationally hard.
We aim to investigate natural game-theoretical and computational questions as future works:
Does a core stable group structure exists for every network? What about other stability concepts? What would be strategies of players to improve popularity?
The proposed games are instances of a more general framework for network-based cooperative games, where payoffs of players are given by various centrality indices. It is interesting to extend the work by considering other centralities, and different forms of social networks (e.g. directed, weighted, signed). Furthermore, one could also explore the evolution of social groups in a dynamic setting.

\begin{figure}
	\centering
	\includegraphics[width=\textwidth]{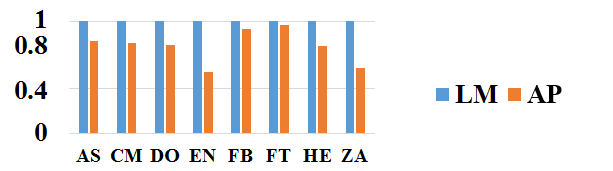}
	\caption{  The percentage of nodes with a higher payoff (compared to the grand coalition) in coalitions generated by the heuristics for real-world networks.}
	\label{fig:increasing-rate}
\end{figure}


\bibliographystyle{splncs03}

\end{document}